\documentclass{extarticle} 
\usepackage{geometry}
\usepackage{times}  
\usepackage{helvet}  
\usepackage{courier}  
\usepackage[hyphens]{url}  
\usepackage{graphicx} 
\usepackage{natbib}  
\usepackage{caption} 
\usepackage{algorithm}
\usepackage{algorithmic}
\usepackage{todonotes}
\usepackage{newfloat}
\usepackage{listings}
\DeclareCaptionStyle{ruled}{labelfont=normalfont,labelsep=colon,strut=off} 
\floatstyle{ruled}
\newfloat{listing}{tb}{lst}{}
\floatname{listing}{Listing}

\usepackage{xurl}
\usepackage{booktabs}

\usepackage{xcolor}
\usepackage{amsfonts,amsmath,amssymb,amsthm}
\usepackage[noabbrev,capitalize]{cleveref}
\newtheorem{theorem}{Theorem}[section]

\newtheorem{proposition}[theorem]{Proposition}

 \newtheorem{example}[theorem]{Example}

\newtheorem{claim}{Claim}
\newtheorem{remark}[theorem]{Remark}
\theoremstyle{definition}

\crefname{table}{Table}{Tables}
\crefname{remark}{Remark}{Remarks}

\usepackage{enumitem}

\usepackage{tikz}
\usetikzlibrary{shapes}
\usepackage{pgfplots}
\usepackage{pgffor}
\usetikzlibrary{patterns}

\usepackage{relsize}
\usepackage{cancel}

\tikzset{fontscale/.style = {font=\relsize{#1}}
}

\usetikzlibrary{backgrounds}
\usetikzlibrary{calc,shapes,decorations,matrix,arrows} 

\usetikzlibrary{calc}
\usetikzlibrary{fit}
\tikzset{
	between/.style args={#1 and #2}{
		at = ($(#1)!0.5!(#2)$)
	}
}

\title{The Cost Perspective of Liquid Democracy: \\ Feasibility and Control}
\usepackage{etoolbox}

\usepackage{blindtext}

\author {
    Shiri Alouf-Heffetz\textsuperscript{\rm 1},
    Łukasz Janeczko\textsuperscript{\rm 2},
    Grzegorz Lisowski\textsuperscript{\rm 2}, 
    Georgios Papasotiropoulos\textsuperscript{\rm 3}  \vspace{1em}\\
    \textsuperscript{\rm 1}\small{Ben Gurion University of the Negev, Beersheba, Israel}\\
    \textsuperscript{\rm 2}\small{AGH Univesity, Kraków, Poland}\\
    \textsuperscript{\rm 3}\small{University of Warsaw, Warsaw, Poland}\\
    {\footnotesize
    shirih@post.bgu.ac.il}\\
    {\footnotesize
    \{ljaneczk,glisowski\}@agh.edu.pl,}\\
		{\small
    gpapasotiropoulos@gmail.com}
}
\date{}

\makeatletter
\newcommand*{\balancecolsandclearpage}{%
  \close@column@grid
  \cleardoublepage
  \twocolumngrid
}

\usepackage{changepage}

\begin{document}

\maketitle

\begin{abstract}
We examine an approval-based model of Liquid Democracy with a budget constraint on voting and delegating costs, aiming to centrally select casting voters ensuring complete representation of the electorate. From a computational complexity perspective, we focus on minimizing overall costs, maintaining short delegation paths, and preventing excessive concentration of voting power. Furthermore, we explore computational aspects of strategic control, specifically, whether external agents can change election components to influence the voting power of certain voters.
\end{abstract}

\section{Introduction}

\label{sec:intro}
In a perfect democracy, every citizen would ideally have the time, knowledge, and means to actively participate in every political decision impacting the society.
This ideal is hindered by various costs, such as the effort required to understand complex issues or the financial burden of frequent decision-making processes. 
Liquid Democracy addresses these challenges by allowing well-informed individuals to vote directly, while others (transitively) delegate their voting rights to trusted representatives. By requiring only a subset of the electorate to form opinions and vote, it reduces overall cognitive and operational costs. 
Originating in the 20$^{\text{th}}$ century ~\citep{tullock,miller1969program,green2015direct,blum2016liquid}, Liquid Democracy has seen widespread implementation in recent years~\citep{paulin2020overview} and has since become a prominent research area within Computational Social Choice ~\citep{grossi2024enabling}.
Despite its focus on reducing election-related costs, their explicit consideration remains unexplored. We address this by analyzing a natural Liquid Democracy framework operating within a budget constraint.

In a Liquid Democracy system costs can arise from various factors, such as the time or effort needed to understand the election topic or to identify participants aligned with their views for vote delegation. These costs broadly reflect preferences between delegation and direct voting. Alternatively, they could correspond to monetary incentives to encourage participation in decisions on complex issues, e.g., when voters prefer not to engage but recognize that some must for the collective good. Costs may also stem from the voters' expertise based on voting history or their profile and interests. Depending on the context, they can be specified by voters or inferred by the system.

Selecting trusted representatives under budget constraints extends beyond Liquid Democracy. The following examples demonstrate how our model and questions apply to other domains where being or trusting a representative incurs costs, aiming to minimize these expenses while meeting explicit trust requirements. 
     In a sensor network, a subset of sensors may be designated as cluster heads to aggregate and transmit data to a central hub \citep{mishra2019trust}. These could be chosen to optimize communication and data-driven decision-making, while minimizing costs related to head designation, such as energy consumption and security concerns.
     In autonomous vehicle fleets, leadership roles in decision-making, like route planning and traffic management, incur costs due to the advanced processing and communication required in control vehicles. Trust from other vehicles is essential for seamless, reliable operation and safety \citep{awan2020stabtrust}. 
     In decentralized blockchain networks, validators or miners validate transactions, ensuring network integrity and reducing computational costs for non-validator nodes \citep{gersbach2022staking}.
     Various implementations in cryptocurrency systems (see, e.g., Project Catalyst and Optimism)
    allow stakeholders to vote on protocol changes through representatives, rewarding the latter for their expertise and time spent understanding protocol details.

 In Liquid Democracy literature, cost is often assumed to exist for differentiating voters' options, though not in a quantifiable manner. For instance, in the works of \citet{escoffier2019convergence,escoffier2020iterative} and \citet{markakis2021approval}
 voters prefer options that are less costly, but this preference is not numerically captured. In contrast, \citet{armstrong2021limited} and \citet{bloembergen2019rational} explicitly consider costs 
 albeit in a decentralized scenario.
 A centralized approach is explored by \citet{birmpas2024reward}, which involves rewarding voters to become representatives. 
While this work, like ours, focuses on selecting a budget-feasible solution, their objective is to maximize the probability of identifying a ground truth, whereas ours emphasizes voter representation.

The model we examine aligns also with research on approval voting for committee elections, especially where voters choose representatives among themselves. Refer to the works of \citet{alon} and \citet{kahng2018ranking} for motivation and related literature. However, staying in line with the context of Liquid Democracy, we assume transitivity of trust, i.e., each voter trusts those whom their trustees trust. Our optimization criteria differ as well.

\subsection*{Our Contribution} In our work we explore decision-making scenarios modeled by a graph where vertices represent entities (or voters, in the terminology of Liquid Democracy) and edges denote trust relationships (potential delegations). 
Each voter is associated with a \emph{voting cost} 
and a \emph{delegating cost}. 
Importantly, our results also apply when delegating costs are zero -- a scenario relevant to several of the examples presented earlier. We work towards selecting voters to cast ballots under a constraint on the costs. Crucially, voters may be asked to vote despite their reluctance not only if they have lower voting costs than others but also if they are widely trusted as representatives. 
In line with theoretical works and common practices in Liquid Democracy, we assume transitivity of trust: If a voter votes on behalf of another, the latter is considered as satisfied as long as there is a directed path of approved representatives connecting them.
To incorporate this 
we require that all voters are represented, either directly or through paths of trust relationships. In the first part of our work, we design and analyze from the perspective of computational complexity mechanisms that:

\vspace{0.2cm}
\begin{adjustwidth}{0.45cm}{0.45cm}{ \quote{\emph{Minimize the total voting and delegating cost, while ensuring that every voter is either selected to vote or has a chain of trust leading to someone who is selected.}}}
\end{adjustwidth}
\vspace{0.2cm}

\noindent On top of this feasibility constraint, still under the objective of cost minimization, and towards suggesting intuitively better solutions, we additionally aim (1) \emph{to ensure each delegating voter has a reasonably short delegation path to a representative}, mitigating diminishing trust with path length, and (2) \emph{to prevent excessive concentration of voting power}, maintaining democratic legitimacy. Both desiderata have been identified as crucial and analyzed in prior works on Liquid Democracy, though in different contexts
\citep{brill2022liquid,golz2021fluid,kling2015voting}.

In the second part of our work we define and algorithmically study a family of control problems related to the aforementioned constraint on voters' voting power. Control of elections is a prevalent area of Computational Social Choice \citep{faliszewski2016control}, recently garnered attention in Liquid Democracy settings \citep{bentert2022won,alouf2024controlling}. Unlike the literature focusing on controlling outcomes, since voters' preferences over them are not explicitly present in our framework, we consider scenarios where:

\vspace{0.2cm}
\begin{adjustwidth}{0.45cm}{0.45cm}{ \quote{\emph{An external agent favoring a certain voter, aims to manipulate the instance to enhance their power.}}}
\end{adjustwidth}
\vspace{0.2cm}

\noindent This marks a conceptual contribution, which could pave the way for further strategic studies in Liquid Democracy and beyond, shifting attention from just election outcomes.

Our results show that if each voter specifies at most one individual as an acceptable representative, optimally selecting voters to cast ballots is computationally feasible under all of the examined constraints. However, allowing voters to approve more representatives makes most of the related problems hard. A similar pattern is observed in the examined control problems.
Despite the apparent similarity in some of the statements of our findings, the proofs themselves vary significantly. Furthermore, we designed the proofs attempting to be broadly applicable, allowing us to extend the results to related computational problems which are well-motivated within the examined setting. 

\medskip
\noindent All missing proofs are provided in the Appendix.

\section{Model and Definitions}
\label{sec:prelims}
The core component of the  setting that we study is a directed graph $G(N,E)$, which we call a \textit{delegation graph}, where the set $N$ of vertices represent \emph{voters} and the set $E$ of edges show to whom a voter might delegate. Specifically, an edge from a voter $i \in N$ to a voter $j\in N$ signifies that $i$ approves $j$ as their (immediate) representative.
We will refer to the tuple $(G(N,E),v,d)$ as an instance of a Liquid Democracy scenario under costs, or a \emph{cLD election} in short.
An important parameter of a cLD election is the \emph{maximum out-degree} of vertices of the corresponding delegation graph, which we denote by $\Delta$. We will call \emph{predecessors} of a vertex $i\in N$ all the vertices that have a directed path to $i$. Moreover, each voter $i\in N$ is associated with two costs: $d(i)$ for delegation and $v(i)$ for voting. We note that both values may be zero and that there is no inherent ordering between them.

A \textit{delegation function} is a mechanism that takes a cLD election $(G(N,E),v,d)$ as input, and returns a set $C\subseteq N$ representing voters who will cast a ballot, together with a set $D\subseteq E$ of exactly one outgoing edge for each vertex not in $C$ denoting the selected delegations.
We call the pair $(C,D)$ a \emph{solution} under a delegation function on $(G(N,E),v,d)$.
Under a delegation function, we refer to the selected subset of voters $C$ as the set of \textit{casting voters} and as \emph{delegating voters} to the voters from $N\setminus C$.
The first computational problem that we consider follows.

\begin{table}[h!]
	\centering
	\begin{tabular}{lp{11.6cm}}  
		\toprule
	 \multicolumn{2}{c}{\textsc{Delegate Reachability} } \\
		\midrule
{\small{\textbf{Input}}} & A cLD election $(G(N,E),v,d)$ and a budget parameter $\beta \in \mathbb{N}^*$.\\
{\small{\textbf{Question}}} & Does there exist of a delegation function of solution $(C,D)$ satisfying:\\ &
(i) the \emph{cost constraint}: $\sum_{i \in C}v(i)+ \sum_{i \notin C}d(i) \leq \beta$,\\ & 
(ii) the \emph{reachability constraint}: for every vertex $i\notin C$ there is a directed path to a vertex in $C$ using only edges in $D$.\\
		\bottomrule
	\end{tabular}
\end{table}
\noindent When $\beta$ is not part of the input of the considered problem, a solution, and in turn a delegation function, is said to be \textit{cost-minimizing} if it satisfies the cost constraint for the minimal value of $\beta$ for which 
reachability is also met. 
 
In principle, a delegation function can be viewed as a way to partition $N$ into casting and delegating voters, though it does not yet ensure that delegation cycles are avoided. However, a solution satisfying the reachability constraint not only involves the specification of casting voters but also designates a specific casting voter, or a \emph{representative}, for each delegating one, by following the paths to casting voters indicated by $D$.
Under a delegation function of the solution $(C,D)$, we denote by $R_j$ the set of delegating voters who have a path to a casting voter $j$ through edges in $D$, i.e., those who will ultimately be represented by $j$. Consequently, $j$ will vote with a voting weight, or \emph{voting power}, of $|R_j| + 1$, representing the voters in $R_j$ along with themselves.

        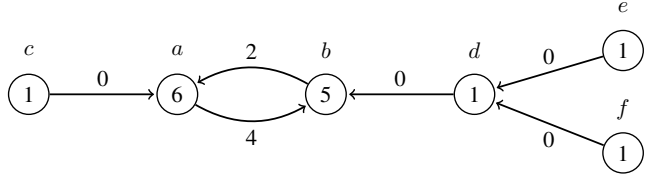
\begin{figure}[t]	
            \centering
            \scalebox{0.85}{
                \begin{tikzpicture}
                [->,shorten >=1pt,auto,node distance=2.3cm,
                semithick]
                \node[shape=circle,draw=black] (A)  {6};
                \node[shape=circle, above of = A, yshift=-4.6em] (A')  {$a$};

                \node[shape=circle,draw=black, right of = A] (B)  {5};
                \node[shape=circle, above of = B, yshift=-4.6em] (B')  {$b$};

                \node[shape=circle,draw=black, left of=A] (C)  {1};
                \node[shape=circle, above of = C, yshift=-4.6em] (C')  {$c$};

                \node[shape=circle,draw=black,  right of=B, shape=circle,draw=black] (D)  {1};
                \node[shape=circle, above of = D, yshift=-4.6em] (D')  {$d$};

                \node[shape=circle,draw=black,  right of=D, above of =D, yshift = -4.6em] (E)  {1};
                \node[shape=circle, above of = E, yshift=-4.6em] (E')  {$e$};

                \node[shape=circle,draw=black,  below of = E, yshift = 2em] (F)  {1};
                \node[shape=circle, above of = F, yshift=-4.6em] (F')  {$f$};
                
                \draw [thick, ->, bend right]
                (B) to node[above] {2}(A);
                \draw [thick,->, bend right] (A) to node[below] {4}  (B) ;

                \draw [thick,->] (C) to node[above] {0} (A) ;
                \draw [thick,->] (D) to node[above] {0} (B) ;
                \draw [thick,->] (E) to  node[above] {0}(D) ;
                \draw [thick,->] (F) to  node[below] {0}(D) ;

                \end{tikzpicture}}

        \caption{An example of a delegation graph on $6$ voters, namely $\{a,b,c,d,e,f\}$. Numbers in vertices indicate voting costs, while those on an edge $(u,v)$ show the delegating cost of voter $u$.}\label{fig:initialEx}    
        \end{figure}
        
\begin{example}\label{ex:initialEx}
Consider the delegation graph in \cref{fig:initialEx}. There are 6 voters, $a$ to $f$, for whom casting a ballot is more expensive than delegating. To ensure that a path exists from each voter to a casting one, any set including $a$ or $b$ could be selected to cast a ballot. However, if cost minimization is the (sole) objective, setting $C=\{a\}$ is necessary and it comes at a total cost of 8: 6 from $a$'s voting cost, 2 from $b$'s delegating cost, while the rest delegate for free. 
Under this solution, $a$ holds a voting power of 6.
\end{example}
\smallskip

In practice, and setting aside the costs for a moment, an electorate in a Liquid Democracy scenario under the discussed model can be partitioned into two groups: Voters who prefer casting a ballot and those who do not. Essentially, a delegation function satisfying the cost and reachability constraints is responsible for determining which voters who feel like they do not want to cast a ballot should be required to do so. This could be due to it being easier for them to become voters rather than others, or for the benefit of society, meaning they are widely approved as representatives by others (either directly or through a path of trust relationships). A natural question then arises regarding what happens to those voters who prefer to cast a ballot rather than delegate.\footnote{Voters who prefer to cast a ballot need not be asked to specify delegating costs at all, however, for uniformity, we describe the setting generally, as this observation does not affect our results.} Importantly, the objective of cost minimization in our setting suggests the expected outcome: voters who are willing to vote will be asked to do so.

\begin{remark}
A cost-minimizing solution includes in the selected set of casting voters every voter who prefers to cast a vote rather than delegate, i.e., satisfies $ \{i \in N: v(i) < d(i)\} \subseteq C$. This also applies to voters who do not express their consent for being represented by others.
\end{remark}

\section{Feasible Delegation Functions}
\label{sec:del_reachability}

We begin our analysis on the computational complexity of finding feasible delegation functions by showing that \textsc{Delegate Reachability} is tractable.

\begin{theorem}
\textsc{Delegate Reachability} is solvable in polynomial time.
\label{del_rep_positive}
\end{theorem}
\begin{proof}[Proof Sketch]
We present a greedy method for computing a cost-minimizing solution. We start by constructing the decomposition $\mathcal{D}(G)$ of the delegation graph $G$ into strongly connected components. 
Recall that a subgraph of a directed graph is called a strongly connected component (SCC) if, for every pair of vertices $u$ and $v$ within it, $u$ is reachable from $v$ and vice versa. The \emph{decomposition} of a graph $G$ into SCCs is a directed acyclic graph having one vertex for every maximal strongly connected component of $G$
and a directed edge from the vertex for component $S$ to the vertex for component $S'$ if $G$ contains an edge from any vertex in $S$ to any vertex in $S'$ (refer to the textbook by \citeauthor{dasgupta2006algorithms} (\citeyear{dasgupta2006algorithms}) for a more extensive discussion on SCC decompositions). We notice that in each feasible solution of \textsc{Delegate Reachability}, it is necessary to select at least one voter as a casting voter from each set of voters that belong to the component corresponding to a sink of $\mathcal{D}(G)$; note that $\mathcal{D}(G)$ has a sink as it is acyclic. 
First, identify all voters whose voting cost is strictly less than their delegation cost. These voters are designated as casting voters.
Then, for each sink of $\mathcal{D}(G)$ that has not yet been assigned at least one casting voter, it is sufficient to nominate the voter who minimizes the sum of their voting cost plus the delegation costs 
of the other voters in the sink.
\end{proof}

In many cases, it may be impossible to find a feasible delegation function that satisfies both the cost and the reachability constraints. Instead of deeming such scenarios indomitable immediately, we suggest first exploring a natural way or circumventing this issue, already addressed in other Liquid Democracy frameworks as well
\citep{colley2023measuring,campbell2022liquid,markakis2021approval}. One way to expand the solution space is to allow for a given number, $\alpha$, of abstainers, requiring feasibility only after excluding them. To expand further, as discussed in \cref{sec:intro}, ideally, in the examined model voters uncomfortable with casting a ballot should be assigned to a casting voter through a path of trust relationships. However, if this proves impossible, and with a slight abuse of notation, a solution would be to permit a small number of voters to remain unrepresented. Unfortunately, with this additional degree of flexibility, the corresponding variant of \textsc{Delegate Reachability} becomes computationally hard.

\begin{theorem}
    Allowing for $\alpha$ abstainers, \textsc{Delegate Reachability} is NP-complete.
    \label{del_rep_neg}
\end{theorem}

Incorporating the possibility of abstainers into the model also aligns well with the line of work on control problems discussed in \cref{sec:intro} and which we analyze in \cref{sec:control}. Imagine a controller (not necessarily a malicious one this time) aiming to alter the election components to turn a negative instance of \textsc{Delegate Reachability} into a positive one, thereby ensuring that the election is not deemed useless. By selecting a set of abstainers to make certain that \textsc{Delegate Reachability} admits a valid solution, the controller (or election organizer) can essentially determine which voters to exclude to make the instance feasible. For completeness, we note that with a slight modification to the proof of \cref{control-hard2} we can also prove hardness for the problem of adding voters to ensure feasibility.

We conclude the section by highlighting that \cref{del_rep_neg} holds for a value of $\alpha$ that depends on the parameters of the input instance. However, if $\alpha$ is fixed, then a generalization of the result of \cref{del_rep_positive} is possible. Specifically, one can easily apply, in polynomial time, the algorithm presented after the deletion of any possible set of at most $\alpha$ voters.

\begin{proposition}
\label{const_alpha_del_reach_in_p}
Allowing for $\alpha\in \mathcal{O}(1)$ abstainers,
    \textsc{Delegate Reachability} is polynomially solvable.
\end{proposition}

In essence, \Cref{del_rep_positive} asserts that the existence of a \emph{sufficiently good} delegation function which meets basic cost and representation requirements can be efficiently decided. In what follows, we will focus on identifying necessary and sufficient conditions for achieving intuitively \emph{significantly better} solutions, referring mainly to sets of casting voters that satisfy additional requirements. 

\subsection{Refined Feasible Delegation Functions}
\label{sec:constrained}

We will focus on instances that admit a feasible solution under \textsc{Delegate Reachability}. Specifically, when sets of casting voters that meet both cost and reachability constraints exist, one delegation function may be seen as superior to another based on various metrics, with two natural requirements being (1) aiming for solutions that create short paths from delegating voters to casting voters, and (2) aiming for solutions that propose casting voters which will not acquire excessively large voting power.

Objective (1) follows the principle that the longer the proposed delegation path from a delegating voter $v$ to a casting voter, the less $v$ trusts their final representative. It is natural to expect (and always assumed in Liquid Democracy) that a voter trusts their immediate successor, and also the successor's successor. However, this plausibly comes with a reduction in trust \citep{boldi,armstrong2024optimizing,brill2022liquid}.
Hence, the shorter the delegation paths, the better. 

Objective (2) has been extensively discussed and analyzed in Liquid Democracy frameworks. 
We refer to \citet{golz2021fluid} and \citet{kahng2021liquid} for comprehensive analyses on the inherent danger that voters with high power can bring upon elections, and the threat this poses to the democratic nature of the system.

\begin{example}
    Take the instance from \cref{fig:initialEx}. Notice that if the budget allows, selecting $\{a,d\}$ as casting voters could be better than choosing $\{a\}$, as it splits the voting power between two casting voters. However, if the goal is to ensure that delegating voters have short distances to their representatives, setting the set of casting voters to $\{b\}$ is better than $\{a\}$, as then all voters would have a path of length at most 2 to their representative, budget permitting.
\end{example}

\noindent We now define the computational problems aimed at refining the solutions of \textsc{Delegate Reachability}.
\begin{table}[H]
	\centering
	\begin{tabular}{lp{11.6cm}}  
		\toprule
	 \multicolumn{2}{c}{\textsc{Bounded Max Length} / \textsc{Bounded Power} } \\
		\midrule
{\small{\textbf{Input}}} & A cLD election $(G(N,E),v,d),$ a budget $\beta \in \mathbb{N}^*$ and a parameter $\ell\in \mathbb{N}^*$.\\
{\small{\textbf{Question}}} & Does there exist a delegation function of solution $(C,D)$ which satisfies the cost and reachability constraints while ensuring that: \\&
$\blacktriangleright$ Each path from a delegating voter to its nearest voter in $C$ via edges from $D$ is of length at most $\ell$, for \textsc{Bounded Max Length} problem,\\& $\blacktriangleright$ The maximal voting power does not exceed $\ell$, for \textsc{Bounded Power} problem.\\
		\bottomrule
	\end{tabular}
\end{table}

Henceforth, we will pay particular attention to delegation graphs with $\Delta=1$. 
Conceptually, this scenario represents the first step away from direct democracy, with $\Delta=0$ where all voters cast their own ballots without representation. In contrast, $\Delta=1$ resembles forms of classic representative democracy, which is of course lacking the transitivity in ballots, while a priori partitioning voters into those who will represent and be represented.
Technically, as we will demonstrate, the examined problems are computationally hard for instances with $\Delta=2$. Therefore, exploring the restricted yet well-motivated case of $\Delta=1$ is a natural first step toward characterizing polynomially solvable instances.

A directed tree, i.e., a directed graph the underlying undirected graph of which has a tree structure, is  \emph{upwards-directed} if there is a vertex that can be specified as its \emph{root} or \emph{sink}, and all other vertices have a directed path leading to it. 
It is easy to see that a delegation graph with $\Delta=1$ consists of a disjoint union of (weakly) connected components, each containing at most one directed cycle. If a cycle exists, its vertices have no outgoing edges towards vertices outside the cycle. Thus, each component can be viewed as a graph with at most one cycle, whose removal would decompose the graph into a forest of upwards-directed trees. With this observation, we obtain the following result through a dynamic programming procedure.

\begin{theorem}
    For $\Delta \leq 1$, 
    \textsc{Bounded Max Length} is solvable in polynomial time.
\label{dp-maxlength}
\end{theorem}

In contrast to \textsc{Delegate Reachability}, moving to instances where $\Delta$ is larger than 1,  even for cases where the maximal out-degree of voters is 2, \textsc{Bounded Max Length} becomes computationally hard.

\begin{theorem}
    For $\Delta>1$, \textsc{Bounded Max Length} is NP-complete for every fixed $\ell\geq 2$.
    \label{length-hard}
\end{theorem}

\begin{proof}[Proof Sketch]
We show that the problem we consider is NP-hard by reducing from \textsc{3-SAT}. We first show the claim for $\ell=2$. Take a 3-CNF formula $\varphi$ with the set of variables $X= \{x_0, \dots, x_m \}$ and the set of clauses $\mathcal{C}=\{C_0, \dots, C_l\}$. We will also represent each clause $C_j$ as $\{L_j^1, L_j^2, L_j^3\}$, where each $L_j^i$ corresponds to a different literal in $C_j$, for $i\in \{1,2,3\}$. Let us construct what we call an \emph{encoding} of $\varphi$. Now, for each clause $C_j\in \mathcal{C}$ we add three voters, i.e., a \emph{clause voter} $v_{C_j}$, as well as two \emph{clause dummy} voters $d_{C_j}, d_{C_j}'$ corresponding to $C_j$. Additionally, for each variable $x_i\in X$, we add two \emph{literal voters} corresponding to $x_i$ and $\neg x_i$, namely voters $v_{x_i}$ and $v_{\neg x_i}$. For every clause $C_j$, we construct edges, i.e., possible delegations, from the clause voter $v_{C_j}$ to the clause-dummy voters $d_{C_j}, d_{C_j}'$. Also there are edges, i.e., possible delegations, from $d_{C_j}$ to the literal voter corresponding to $L_j^1$ and from $d_{C_j}'$ to the literal voters that correspond to $L_j^2$ and to $L_j^3$. Finally, for every $x_i\in X,$ we add the following pair of edges: from $v_{x_i}$ to $v_{\neg x_i}$ and from $v_{\neg x_i}$ to $v_{x_i}$. To complete the construction, we set voting costs to $1$ and delegating costs to $0$ and let $\beta=|X|$ and $\ell=2$. Note that it holds $\Delta=2$ and that the encoding of $\varphi$ consists of $2|X| + 3|\mathcal{C}|$ voters and the shortest path from each clause voter to a literal voter is of length exactly 2.

Assume $\varphi$ is satisfiable. Then, for each pair of literal voters $v_{x_i}$ and $v_{\neg x_i}$, we take $v_{x_i}$ as a casting voter if $x_i$ is true in the considering satisfying assignment of $\varphi$, and $v_{\neg x_i}$ otherwise. This selection is budget feasible and ensures that each clause voter is at a distance of exactly two from a casting voter. Dummy and literal voters are also at a distance of at most two each. 
Conversely, consider the case that $\varphi$ is unsatisfiable. It holds that any feasible solution for the encoding of $\varphi$ requires choosing exactly one from each pair of literal voters, due to the budget constraint, and also that only literal voters can be casting voters. This leads to some clause voters being three steps away from a casting voter, violating the path length constraint.

We note that the construction can be modified to work for any larger value of $\ell,$ by increasing the distance between clause and literal voters, from $2$ to $\ell$.
\end{proof}

We next examine the problem of selecting representatives with the constraint being an upper bound on the maximal voting power of casting voters. While the obtained results are analogous to those for \textsc{Bounded Max Length} (being polynomially solvable for $\Delta\leq1$ and NP-complete otherwise), the specifics of the proofs significantly differ.

\begin{theorem}
\label{dp-power}
    For $\Delta \leq 1$, \textsc{Bounded Power} is solvable in polynomial time.
\end{theorem}

\begin{proof}[Proof Sketch]
Note that in the case of $\Delta=0$ the problem is trivial. We show that the claim holds by providing a dynamic programming algorithm. For this we define $dp[v,i,k]$ as the minimum cost of a (partial) solution that gives a voting power upper bounded by $k$ to $v$ and upper bounded by $\ell$ to all other casting voters of the considered subgraph, among the following vertices in $G$: $v$, its first $i$ incoming neighbors and all of their predecessors. 
This is justified by the fact that vertices that will be considered next can increase the voting power of $v$ but not the voting power of voters that have been already clasified as casting. We will compute a value for $dp[v,i+1,k]$ based on $dp[v,i,\cdot]$ and $dp[v_{i+1},p(v_{i+1}),\cdot]$, where $v_{i+1}$ is the $(i+1)^{\text{th}}$ in-neighbor of $v$ and $p(v_{i+1})$ is the number of its in-neighbors. If $G$ is an upwards-directed tree, then the voter corresponding to the root should be nominated as a casting voter, among others, and a similar argument can be stated for arbitrary graphs satisfying $\Delta=1$. 
The values of the table can be computed by a bottom-up approach. Observe that the size of the table is polynomial in the input and it can be shown that the computation of its values can be also done efficiently.
\end{proof}

As in the \textsc{Bounded Max Length}, transitioning from a single accepted delegate per voter to allowing the approval of more potential representatives significantly increases the complexity of \textsc{Bounded Power}.

\begin{theorem}
    For $\Delta>1$, \textsc{Bounded Power} is NP-complete for every fixed $\ell\geq 4$.
\label{power-hard}
\end{theorem}

\begin{proof}[Proof Sketch]
For clarity, we first outline the proof for the case of $\ell=4$. We prove NP-hardness by reducing from \textsc{Vertex Cover} on 3-regular graphs. Considering such an instance $\mathcal{I}$ on $\nu$ vertices, we create the following encoding of $\mathcal{I}$ into an instance of \textsc{Bounded Power} problem: 
First, for every vertex of $\mathcal{I}$, we construct a \emph{vertex voter} corresponding to $v$. Moreover, for every edge $e$ of $\mathcal{I}$ we construct an \emph{edge voter} corresponding to $e$. Furthermore, if an edge $e$ of $\mathcal{I}$ is between vertices $i,j$, we construct edges, i.e., possible delegations, from the voter corresponding to $e$ to the voters who correspond to $i$ and to $j$. We finally add a set $D$ of $\nu$ \emph{dummy voters}, who will vote at a cost of $0$. For each vertex voter we add an edge to exactly one dummy voter in a way that each voter from $D$ has exactly one incoming edge from vertex voters. Finally, for each voter $v$ in $D$ we add two additional voters, who are willing to delegate to $v$. We set the rest voting costs to $1$ and all delegating costs to $0$ and also we set $\beta=k,$ where $k$ is the decision parameter in $\mathcal{I}$, and $\ell=4$. Note that $\Delta=2$.

If $\mathcal{I}$ is a positive instance, we construct a budget feasible delegation function having the vertex voters corresponding to vertices in the cover along with dummy voters as casting voters. Then, each edge voter is represented by a vertex voter selected to cast a ballot, while the non-selected vertex voters will be represented by the dummy voters. Voting power of casting voters is bounded by 4. 
Supposing that \(\mathcal{I}\) is a negative instance, a selection of $k$ vertex voters as casting is not sufficient to cover all edge voters. This leads to dummy voters having to represent some edge voters, which causes the voting power to exceed the bound of \(\ell = 4\). Hence, no feasible delegation function exists. In turn, the encoding of \(\mathcal{I}\) is a negative instance of \textsc{Bounded Power} as well. 
 
Increasing the number of voters approving each dummy voter from $2$ to $\ell-2$, and applying the same arguments, proves the theorem for any value of $\ell$.
\end{proof}

A further natural restriction on delegation functions involves averaging the lengths of paths to casting voters. While \textsc{Bounded Max Length} can be seen as having an egalitarian constraint, restricting the \emph{sum of path lengths} is a utilitarian one. This new restriction balances \textsc{Bounded Max Length} and \textsc{Bounded Power}, allowing a casting voter to represent either multiple voters nearby or a few farther away. The proof of \cref{power-hard} immediately establishes the hardness of \textsc{Bounded Sum Length}, which asks whether a delegation function of solution $(C,D)$ can satisfy cost and reachability constraints while keeping the sum of path lengths via $D$ to each casting voter within a given bound $\ell$. Moreover, \cref{dp-maxlength} can be easily adapted to address this restriction.

\begin{proposition}
        For $\Delta>1$, \textsc{Bounded Sum Length} is NP-complete for every fixed $\ell\geq 4$, but for $\Delta \leq 1$ it is solvable in polynomial time.
\end{proposition}

\section{Strategic Control}
\label{sec:control}

We now begin our exploration of algorithmic questions concerning strategic control by an external (malicious) agent, referred to as the \emph{controller}. In classic literature of election control (refer for instance to the work of \citeauthor{faliszewski2016control} (\citeyear{faliszewski2016control})), the controller is typically considered able to manipulate components of the election such as the set of voters or candidates to achieve a desired outcome, under a prespecified voting rule. However, in our model presented in \cref{sec:prelims}, which aligns with most studies in Liquid Democracy, we have abstracted away from the actual voting procedure and our focus has been on determining which voters will cast a ballot and consequently on the extent of their voting power; without considering the specifics of their votes.
Therefore, unlike traditional election control studies and questions where the controller attempts to influence the final outcome, in our work the controller cannot take advantage of voters' preferences or final ballots, as these details are beyond our scope. Instead, the controller is allowed to influence only the delegation process and the specification of casting voters and their power. 

We introduce a new class of control problems that integrates seamlessly with our model: The controller aims to manipulate the delegation process toward affecting the voting power of a designated voter. We capture this through the concept of a \emph{super-voter}, defined as the casting voter with the highest voting power, after applying a specified delegation function.\footnote{The term has been used in earlier works (e.g., by \citeauthor{kling2015voting} (\citeyear{kling2015voting})) to refer generally to voters with a large share of incoming delegations, which is (slightly) different from how we use it here.} We assume that the controller does not ex-ante know the precise delegation function that will be used, but only that it will be cost-minimizing while satisfying the reachability constraint, remaining consistent with our work's focus. Consequently, we say that the controller's goal is to ensure that their preferred voter becomes the (sole) super-voter, under every possible delegation function that meets these conditions.

\begin{table}[h!]
	\centering
	\begin{tabular}{lp{11.6cm}}  
		\toprule
	 \multicolumn{2}{c}{\textsc{Control by Adding/Deleting Voters (cav/cdv)}} \\
		\midrule
{\small{\textbf{Input}}} & A cLD election $(G(N,E),v,d),$ a budget $\beta \in \mathbb{N}^*$, a parameter $k \in \mathbb{N}^*$ and\\ & 
$\blacktriangleright$ a partition of $N$ into registered $(N_r)$ and unregistered $(N_u)$ voters, and a designated voter $x \in N_r$, for \textsc{cav}.\\ &
$\blacktriangleright$ a designated voter $x \in N$, for \textsc{cdv}.\\ 
{\small{\textbf{Question}}} & Does there exist a way to add up to $k$ voters from $N_u$, in \textsc{cav}, or to delete up to $k$ voters, in \textsc{cdv}, so that $x$ is the (sole) super-voter in the resulting election, under every cost-minimizing delegation function which satisfies the reachability constraint? \\
\bottomrule
	\end{tabular}
\end{table}

\begin{example}
Take the instance from \cref{fig:initialEx}. There, the cost-minimizing solution selects $a$ as a casting voter, but a malicious external agent preferring $b$ as the super-voter could achieve this by deleting $a$. Suppose that there is also a set $X$ of unregistered voters (not depicted in \cref{fig:initialEx}). Regardless of their costs or which voters they approve of,
the controller cannot make $b$ the unique super-voter under every feasible cost-minimizing solution by adding any subset of $X$.
\end{example}

\begin{theorem}
    \textsc{Control by Adding Voters} is NP-hard.    \label{control-hard1}\end{theorem}
We now move to the case where the controller is able to remove some voters from participating in the election.
As in the case of adding voters, the relevant computational problem is NP-hard, although it requires  different reduction.

\begin{theorem}
    \textsc{Control by Deleting Voters} is NP-hard.\label{control-hard2} 
\end{theorem}

\begin{proof}[Proof Sketch]
We show the hardness of \textsc{Control by Deletion} by reducing from \textsc{Clique}. Given an instance of \textsc{Clique} $(G,k)$ we create an instance $\mathcal{I'}$ of the considered election control problem that contains the set of voters $E$ containing one voter of delegating cost equal to $0$ and voting cost equal to $1$, for each edge in $G$.
Also, we have a set of voters $V$ containing one voter of delegating cost equal to $0$ and voting cost equal to $1$, for each vertex in $G$. 
 Then, we take the set $D$ of $(n-k)+(m-\binom{k}{2})+1$ voters, each of delegating cost equal to $0$ and voting cost equal to $1$ and special voters $x$, being the designated by the controller voter and $y_1,y_2,\dots,y_k,y_{k+1}$, with voting cost equal to $0$ and delegating cost equal to $1$. 
Then, for each edge $e$ that is incident to vertices $u$ and $v$ in $G$, we add two possible delegations: one from the voter of $E$ that corresponds to $e$ towards the voter of $V$ that corresponds to $v$ and one from the voter of $E$ that corresponds to $e$ towards the voter of $V$ that corresponds to $u$. Then, there is an edge from each voter in $V$ to every voter in $Y= \{y_1,y_2,\dots,y_{k+1}\}$ and an edge from each voter in $D$ to voter $x$. Finally say that the upper bound on the number of voters that can be deleted equals $k$.

If there exists a $k$-clique $\mathcal{C}$ in $G$, then deleting voters corresponding to vertices in $\mathcal{C}$ ensures that the voters corresponding to the edges in $\mathcal{C}$ do not delegate their votes to any voter from $Y$ in any cost-minimizing delegation function. Hence, $x$ becomes a super-voter. Conversely, if a $k$-clique does not exist in $G$, then after deleting any $k$ voters, sufficiently many voters corresponding to edges of $G$ will delegate to a voter from $Y$ under some delegation function. Hence, $x$ cannot be guaranteed
to be the only super-voter.
\end{proof}

Other strategic control problems can be defined in a similar vein. For instance, one might focus on adding or removing edges (i.e., possible delegations), instead of voters, to ensure that the preferred voter becomes the sole super-voter. The proofs from \cref{control-hard1,control-hard2} extend to these problems with only minor adjustments. Furthermore, aligning with prior research on strategic control of classic elections    \citep{hemaspaandra2007anyone,faliszewski2011multimode}, one might also explore the \emph{destructive} counterparts of \textsc{cav} and \textsc{cdv}. Specifically, this involves performing alterations to the voters' set towards \emph{preventing} a specific voter from being the unique super-voter under every cost-minimizing delegation function that satisfies the reachability constraint. For these problems the proofs of \cref{control-hard1,control-hard2} apply directly.

\begin{proposition}
\label{rem:str} 
The variants of \textsc{Control by Adding/Deleting Voters} where the controller can add or delete edges of the delegation graph are NP-hard. Moreover, the destructive variants of \textsc{Control by Adding/Deleting Voters} are NP-hard as well.
\end{proposition}

Following our approach in \cref{sec:constrained}, we now turn to the examination of cLD elections where the maximal out-degree of vertices in the input delegation graph is upper bounded by $1$. As one would expect, polynomial solvability now holds, but interestingly, the procedures we propose for solving the control problems differ completely from those used for proving \cref{dp-maxlength,dp-power}.

\begin{theorem}
    \textsc{Control by Adding Voters} is solvable in polynomial time if $\Delta \leq 1$.
    \label{control_pos1}
\end{theorem}
\begin{proof}[Proof Sketch]
For $\Delta = 0$ it suffices to see that no voter may delegate, so $x$ can be the unique super-voter only if $|N| = 1$.
For $\Delta = 1$, we must first verify that \( x \) is a casting voter in every cost-minimizing delegation function; if not, additional voters cannot alter this, meaning the instance has no feasible solution. This verification can be done in polynomial time. We first convert $x$ into the root of a tree by removing its outgoing edges, if any and only if the voting cost of $x$ is higher than its delegating cost. Otherwise, there could be a delegation function where $x$ delegates. We then aim to increase the voting power of $x$ by adding unregistered voters who will delegate to $x$ in any cost-minimizing delegation function. These are identified based on their costs as well as the existence of paths to $x$ or to registered voters who will definitely delegate to $x$. We use a greedy algorithm to add these voters in layers: starting with those directly connected to $x$ or to voters that $x$ will represent, and iteratively including more until $k$ is reached.
\end{proof}

A greedy strategy can also identify which voters can be deleted to ensure a preferred voter becomes the only super-voter under any cost-minimizing delegation function.

\begin{theorem}
    \textsc{Control by Deleting Voters} is solvable in polynomial time if $\Delta \leq 1$. \label{control:pos2}
\end{theorem}

Continuing from \cref{rem:str} our discussion on strategic control problems involving the alteration of possible delegations, we note that the positive results for instances of $\Delta \leq 1$ are also applicable here. For the addition of edges, the proof is analogous to the proof of \cref{control_pos1}. For the deletion of edges, 
one must consider each potential casting voter \( y \) that has a voting power at least as high as that of the designated voter $x$, under at least one cost-minimizing delegation function 
and then delete iteratively the incoming edges to \( y \), in descending order of the number of voters who would lose their paths to \( y \) as a result. This continues until \( y \)’s voting power is reduced to below that of \( x \).

\begin{proposition}
The variants of \textsc{Control by Adding/Deleting Voters} where the controller
can add or delete edges of the delegation graph are solvable in polynomial time if $\Delta \leq 1$.
\end{proposition}

In this section we focused exclusively on controlling the power of casting voters, however, the power of delegating voters should not be underestimated. As highlighted by \citet{zhang2021power} and \citet{behrens2021temporal}, delegating voters can also control a substantial number of votes, impacting the election outcome. For example, a casting voter with direct delegations from $p-1$ voters is intuitively more powerful than one whose voting power of $p$ is achieved through an intermediary receiving exactly $p-2$ delegations, as the latter's power depends entirely on the intermediary. Problems of ensuring a preferred voter gains multiple delegations, even without casting a ballot, naturally arise. Our results extend to such scenarios as well, with proofs being analogous.

\section{Conclusion}
Our work focused on computational problems related to determining feasible delegation functions and controlling the voting power of participants. Specifically, we first addressed questions around the existence of sufficiently good delegation functions (under cost and reachability constraints) considering well-established desiderata from the Liquid Democracy literature. We then concentrated on controlling the voting power of a preferred casting voter by adding or deleting voters to achieve the controller's goal. 
Our results provide a comprehensive understanding of the tractability landscape of these problems' families. 
For most of the considered problems, tractability is characterized (unless P=NP) by allowing each voter to approve at most one other voter as a potential representative.

Even beyond the Computational Social Choice context, the related graph-theoretic problems are of significant interest, making them a promising algorithmic direction
with respect to, e.g., approximation or parameterized algorithms. Our results could serve as a starting point for analyzing the parameterized complexity of the examined problems.
Indicatively, the hardness results from \cref{sec:constrained} show para-NP-hardness for \(\ell\) and the positive results can be extended to parameterize by the number of voters with an out-degree greater than one and their maximum out-degree. \cref{const_alpha_del_reach_in_p} is also an XP[$\alpha$] algorithm, and \cref{control-hard2} is also a W[1]-hardness result for the number of voters that can be deleted. We also underline that our hardness results hold for restricted classes of delegation graphs, such as layered directed graphs with a few layers or even directed bipartite graphs. 
Moreover, the results on \textsc{cav} and \textsc{cdv} hold even under additional restrictions on path lengths or voting power, along with the reachability constraint, in the formulation of these problems.

Our results can be complemented by simulating how factors like path length, voting power bounds, and budget constraints impact the frequency of feasible solutions or how one parameter influences others. 
Exploring how edge existence probabilities, cost distributions, or the number of voters impact feasibility in synthetic datasets could also be insightful. 
Similarly, considering costs, direct voting is democratically optimal but expensive, while (unconstrained) \textsc{Delegate Reachability} offers much cheaper solutions. Solutions with constraints, such as bounded voting power or lengths of delegation paths, fall in the middle. 
Comparing these costs could reveal  the savings of Liquid Democracy over direct voting, along with the additional costs introduced by constraints aimed at yielding intuitively better solutions.

\vfill

\section*{Acknowledgements}
This project has received funding from the European Research Council (ERC) under the European Union’s Horizon 2020 research and innovation programme (grant agreement No 101002854). Georgios Papasotiro\-poulos is supported by the European Union (ERC, PRO-DEMOCRATIC, 101076570). Views and opinions expressed are however those of the author(s) only and do not necessarily reflect those of the European Union or the European Research Council. Neither the European Union nor the granting authority can be held responsible for them. The authors thank Piotr Faliszewski and Piotr Skowron for their helpful feedback and discussions.
\begin{center}
  \includegraphics[width=5cm]{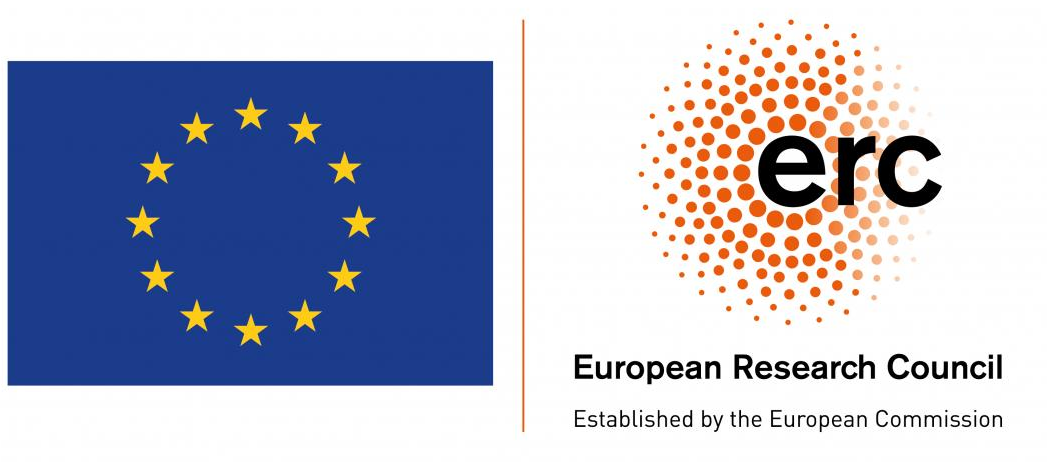}
\end{center}

\bibliographystyle{abbrvnat}
\bibliography{sample1}

\newpage
\noindent {\huge{\textbf{Appendix}}}
\section*{Missing Proofs from \Cref{sec:del_reachability}}

\subsection{Proof of \Cref{del_rep_positive}}

We will show that the claim holds by presenting a greedy method for computing a cost-minimizing solution. We start by constructing the decomposition of the delegation graph into strongly connected components. 

Recall that a subgraph of a directed graph is strongly connected if, for every pair of vertices $u$ and $v$ within it, $u$ is reachable from $v$ and vice versa. The \emph{decomposition} of a graph $G$ into strongly connected components is a new directed graph, $\mathcal{D}(G)$, which has one vertex for every maximal strongly connected component of $G$. There is a directed edge between two vertices in $\mathcal{D}(G)$ if there is an edge from any vertex in one strongly connected component to a vertex in the other in $G$. Such a decomposition can be constructed by classic algorithms in linear time (see e.g. \citep{clrs}) and is acyclic. We call the vertices of the decomposition having zero out-degree \emph{sinks}. Each sink corresponds to a set of vertices in $G$ that are mutually reachable but have no path to vertices outside of the sink in $\mathcal{D}(G)$.

Given a delegation graph $G$, any set of casting voters of a feasible solution for \textsc{Delegate Reachability} contains at least one vertex from each of the sinks of the decomposition $\mathcal{D}(G)$. If this were not the case, the members of this sink in the decomposition would not reach any casting voter, as those vertices do not have a path to the outside of the sink. Moreover, selecting one vertex in a sink component is sufficient to satisfy its members due to the (strong) connectivity property. For cost-minimizing delegation functions, note that when no voter has a voting cost less than the delegating cost, selecting exactly one vertex from sinks is the optimal choice.
    
We proceed by selecting representatives for each sink, but, first, we identify all voters whose voting cost is strictly less than their delegating cost and designate them as casting voters. This step is justified by the fact that choosing these voters as casting voters is always part of the cost-minimizing solution. For each sink $S$, if we have already chosen at least one voter from $S$ as a casting voter in the aforementioned step, then all voters in $S$, as well as their predecessors, already have a feasible path to a casting voter. If $S$ does not yet have a representative, we need to select a voter from it as a casting one. As all remaining voters have a delegating cost that is no more than their voting cost, a cost-minimizing solution includes exactly one voter as a casting one. For a voter $v$, the total cost incurred by choosing $v$ as a casting voter equals the sum of $v$'s voting cost plus the delegating costs of the other voters in the component. Consequently, for each sink of the decomposition that has not yet been assigned a casting voter, we nominate the voter who minimizes the sum of their voting cost plus the delegating costs of the other voters in the sink, as the delegating cost of voters of the component, not belonging to the sink will appear when considering any of the voters in the sink as casting. It is routine to check that the obtained solution is feasible and that it minimizes the total cost.

\subsection{Proof of \Cref{del_rep_neg}}
We prove hardness by a reduction from \textsc{Vertex Cover}. Consider an instance $\mathcal{I}$ of \textsc{Vertex Cover} with a given graph $G'$ with vertex set $V$ and an integer $k$. Let us define what we call an \emph{encoding} of $\mathcal{I}$. First, for every vertex $v \in V$, we construct a \emph{vertex voter} corresponding to $v$. Additionally, for every edge $e$ in $G'$, we construct an \emph{edge voter} corresponding to $e$. If an edge $e$ in $G'$ is between vertices $i$ and $j$, we construct edges, i.e., possible delegations, from the voter corresponding to $e$ to the voters corresponding to $i$ and $j$.
Then, we set the voting costs of all of the voters to $1$ and we set their delegating costs to $0$. Finally, we set $\beta=k$ and $\alpha=|V|-k$.

Suppose that $\mathcal{I}$ is a positive instance of \textsc{Vertex Cover}. Then, we can select $k$ vertices of $V$ that include the endpoints of all the edges in $G'$. Consequently, we can select the vertex voters corresponding to these vertices as casting voters, satisfying the budget constraint of $k$. This ensures that all the edge voters reach one of the selected casting voters. Under the solution implied by the described set of casting voters, at most $|V| - k$ vertex voters remain unrepresented, making the proposed solution feasible due to the specified value of $\alpha$.

Now, suppose $\mathcal{I}$ is a negative instance of \textsc{Vertex Cover} and, towards a contradiction, assume there is a solution to the encoding of $\mathcal{I}$ that leaves at most $|V| - k$ voters unrepresented while having at most $k$ casting voters, denoted as set $S$. Note that in such a delegation function, it is without loss of generality to assume that all casting voters belong to the set of vertex voters of the encoding of $\mathcal{I}$. 

To prove this, consider a solution that includes edge voters as casting voters. We can then form a new delegation function by forcing each of these edge voters to delegate to one of their out-neighbors, making the neighbor a casting voter instead (if they are not already). This adjusted solution remains feasible and has a cost no greater than the original, as both vertex and edge voters share the same voting cost.

Continuing with the reverse direction of the reduction, since $\mathcal{I}$ is a negative instance of \textsc{Vertex Cover}, there is at least one edge voter, say $v$, who cannot delegate to any in $S$. Thus, we have at least $|V| - k + 1$ voters in the encoding of $\mathcal{I}$ who are unrepresented in the suggested solution—these include the $|V| - k$ vertex voters not selected, plus $v$.

\subsection{Proof of \Cref{dp-maxlength}}
Suppose that the delegation graph $G$ is connected; we will discuss the case where $G$ is not connected later in the proof. Say also that the input graph is an upwards-directed tree. If it is not, then by the fact that $\Delta=1,$ the graph has exactly one cycle, the vertices of which might have incoming edges but no outgoing ones. We then say that the cycle appears as the sink of an upwards-directed tree. Since the cycle appears as a sink, at least one vertex of it should be selected as a casting voter. By this inclusion, we can safely delete the outgoing edge of this vertex, as the corresponding voter will not delegate. But then we just end up with an upwards-directed tree. Executing the following procedure, which applies to upwards-directed trees, for every possible choice of a single vertex from the cycle that is a sink, suffices to solve the problem for general connected delegation graphs of $\Delta=1$.

For a given upwards-directed tree $G$, we denote by $T_v$ the induced (upwards-directed) subgraph of $G$ that is rooted at vertex $v$ and contains all (direct and indirect) predecessors of $v$ in $G$.
Also we denote by $T_v^{i}$, the induced (upwards-directed) subgraph of $T_v$ that is rooted at $v$ and includes all vertices among the first $i$ in-coming neighbors of $v$ and any of their predecessors (but none of the in-coming neighbors of $v$ after the $i$-th, neither any of their predecessors). This definition holds for $i=0$ if $v$ is a leaf and for $i\in \{1,2,\dots,p(v)\}$ if not, where $p(v)$ the in-degree of $v$, i.e. the number of direct predecessors. We denote by $v_{i}$ the $i$-th in-coming neighbor of vertex $v$. At what follows we define a recursive dynamic programming algorithm. 

Let $dp[v,i,k]$ represent the minimum cost for selecting casting voters among the vertices in $T_v^i$ in a way that each non-selected vertex, i.e., delegating voter, in $T_v^i$ has a path of length of at most $\ell$ to a casting voter from $T_v^i\setminus\{v\}$ or a path of length of at most $k$ to $v$. Note that definitely $v$ should be selected as a casting voter to ensure feasibility since it is the root of the considered subtree. The minimum cost value regarding $G$ lies in $dp[x,p(x),\ell],$ where $x$ corresponds to the root of the delegation graph $G$. Note that we have assumed that $G$ is connected; if it is not, we must ensure feasibility at each connected component, and the cost of the cost-minimizing solution equals $\sum_{x\in R(G)}dp[x,p(x),\ell],$ where $R(G)$ corresponds to the set of roots of the connected components of $G$, where those components can be considered to be upwards-directed trees.

We will use a bottom-up approach to fill the table. Consider a vertex $v$. If $v$ is a leaf then we set $dp[v,0,k]=v(v)$, i.e., equal to the voting cost of the corresponding voter, for any possible $k$. This is because the subtree of $G$ rooted at a leaf vertex $v$ can only be feasible if $v$ is selected as a casting voter, and then, trivially, all vertices of the subtree (the empty subtree) do have a path of length at most $k$, for any positive $k$, to $v$.

Now say that $v$ isn't a leaf and consider first the case where $i=1$. Then we must pay $v(v)$ to make $v$ a casting voter, since it is a root of the considering subtree. Apart from that, there are two options when considering $T_v^1$: either $v_1$ will delegate to $v$ or $v_1$ will vote directly as well. 
In the first case, we must also pay $d(v_1)$ while ensuring that the maximal path in $T_{v_1}$ is of length at most $k-1$, which leads to paths of length at most $k$ in $T_v^1$. This quantity corresponds to $d(v_1)+dp[v_1,p(v_1),k-1],$ but, since $v_1$ has been counted as a casting voter when computing the value of $dp[v_1,p(v_1),k-1]$ we now also need to subtract the cost $v(v_1),$ from it since we are considering the case where $v_1$ delegated (and that's why we added $d(v_1)$ instead). All in all we get $v(v)+d(v_1)+dp[v_1,p(v_1),k-1]-v(v_1)$. Regarding the second case, we only need to pay $v(v)$ while also ensuring feasibility in $T_{v_1}$, so we get $v(v)+dp[v_1,p(v_1),\ell],$ as the voting cost of $v_1$ has been already counted in $dp[v_1,p(v_1),\ell].$ Among the two options we set $dp[v,i,k]$ to be equal to the value of the one achieving the minimal value while being feasible.

We now move to the case where $i>1$. Note that $v$ has already been considered as a casting voter when computing $dp[v,i-1,k]$ for some $k$. But, again there are two options based on what $v_{i+1}$ will be asked to do, when considering $T_{v_{i+1}}$: either $v_{i+1}$ will delegate to $v$ or $v_{i+1}$ will also vote directly. In the first case, we also need to pay $d(v_{i+1})$ while also ensuring that the maximal path in $T_{v_{i+1}}$ is of length at most $k-1$, which leads to paths of length at most $k$ in $T_{v_{i+1}}$. As before, this quantity corresponds to $dp[v_{i+1},p(v_{i+1}),k-1],$ but since $v_{i+1}$ has been counted as a casting voter there we also need to subtract the cost $v(v_{i+1})$. Additionally, we need to ensure feasibility in the subgraph $T_v^i$ so we add the quantity $dp[v,i,\ell]$.
All in all we get $d(v_{i+1})-v(v_{i+1})+dp[v,i,\ell]+dp[v_{i+1},p({i+1}),k-1]$. In the second case, we also set $v_{i+1}$ as a casting voter. This means that we now have the expression $dp[v_{i+1},p(v_{i+1}),k]+dp[v,i,\ell].$ Among the expressions created for the two discussed cases, again, we set $dp[v,i,k]$ to be equal to the feasible one achieving the minimum value among them. Finally, using a simple backtracking procedure, from the cost of the cost-minimizing solution we can construct the solution itself. The fact that the suggested procedure is polynomial in time is immediate.

\subsection{Proof of \Cref{length-hard}}
We show that the problem we consider is NP-hard by reduction from \textsc{3-SAT}. Take a 3-CNF formula $\varphi$ with the set of variables $X= \{x_0, \dots, x_m \}$ and the set of clauses $\mathcal{C}=\{C_0, \dots, C_m\}$. We will also represent each clause $C_j$ as $\{L_j^1, L_j^2, L_j^3\}$, where each $L_j^i$ corresponds to a different literal in $C_j$. Let us construct what we call an \emph{encoding} of $\varphi$. Now, for each clause $C_j\in \mathcal{C}$ we add three voters, i.e., a \emph{clause voter} $v_{C_j}$, as well as two \emph{clause dummy} voters $d_{C_j}, d_{C_j}'$ corresponding to $C_j$. Additionally, for each variable $x_i\in X$, we add two \emph{literal voters} corresponding to $x_i$ and $\neg x_i$, namely voters $v_{x_i}$ and $v_{\neg x_i}$. Moreover, for every clause $C_j$, we construct edges, i.e., possible delegations, from the clause voter $v_{C_j}$ to clause-dummy voters $d_{C_j}, d_{C_j}'$. Also there are edges, i.e., possible delegations, from $d_{C_j}$ to the literal voter corresponding to $L_j^1$ and from $d_{C_j}'$ to the literal voters that correspond to $L_j^2$ and to $L_j^3$. Finally, for every $x_i\in X,$ we add the following pair of edges: one from $v_{x_i}$ to $v_{\neg x_i}$ and from $v_{\neg x_i}$ to $v_{x_i}$. To complete the construction of the encoding of $\varphi$, we set voting costs to $1$ and delegating costs to $0$ and let $\beta=|X|$ and $\ell=2$. Note that in this case $\Delta=2$. Furthermore, the encoding of $\varphi$ consists of $2|X| + 3|\mathcal{C}|$ voters and the shortest path from each clause voter to a literal voter is of length exactly 2.

Suppose first that $\varphi$ is satisfiable and take a satisfying valuation $S$ over $X$. Then, for every pair of literal voters $v_{x_i}, v_{\neg x_i}$ we select $v_{x_i}$ as a casting voter if $x_i$ is true in $S$, and $v_{\neg x_i}$ otherwise. It is easy to see that this selection is budget feasible. Observe also that then each clause voter is at a distance of exactly two from a literal voter that was selected to cast a ballot. Moreover, dummy voters are also at a distance of at most two from casting voters, and literal voters are at a distance of at most one. It is then immediate that the suggested delegation function is feasible as it allows for paths to a casting voter of length at most two. As a result, the encoding of $\varphi$ is a positive instance as well and this concludes the proof of the forward direction of the reduction.

Suppose now that $\varphi$ is not satisfiable. Regarding its encoding, we notice that by construction exactly one of each pair of literal voters $x_{x_i}, v_{\neg x_i}$ needs to be selected in a feasible solution as otherwise at least one of them would not have a path to a casting voter. Additionally, due to the budget constraint, any feasible solution will include casting voters coming only from the set of literal voters. Consider then the set of casting voters of such a solution and notice that it corresponds to a valuation over $X$ that does not satisfy $\varphi$. But then, by construction, some clause voter is at a distance three from the closest casting voter, which contradicts the feasibility of the solution regarding the specified maximum length of a path to a casting voter.

\subsection{Proof of \Cref{dp-power}}
    The proof that follows concerns the case where the delegation graph $G$ is an upwards-directed tree; the generalization to any graph, under the constraint of $\Delta=1$, is identical to the analogous generalization appearing in the proof of \Cref{dp-maxlength}.

Let $dp[v,i,k]$ correspond to the minimum cost of creating a solution among the following vertices in $G$: $v$, its first $i$ in-coming neighbors and all of their predecessors---which could be seen as the connected component of $v$ after the deletion of edges between $v$ and its out-neighbors as well as its in-neighbors after the $i$-th. This solution is assumed to give a voting power upper bounded by $k$ to $v$ and upper bounded by $\ell$ to all other casting voters of the considered subgraph. This is justified by the fact that vertices that will be considered afterwards can increase the voting power of $v$ but not the voting power of any other vertex that has been already specified as a casting voter. We will compute a value for $dp[v,i+1,k]$ based on $dp[v,i,\cdot]$ and $dp[v_{i+1},p(v_{i+1}),\cdot]$, where $v_{i+1}$ is the $(i+1)$-th in-neighbor of $v$ and $p(v_{i+1})$ is the number of its in-neighbors. 
Regarding the size of the table $dp,$ note that for any fixed $v$, the size of the second dimension of $dp$ is upper bounded by $p(v)$ (and lower bounded by $0$ if $v$ is a leaf or $1$ otherwise) and the size of the third dimension is upper bounded by $\ell\leq n$ (and lower bounded by $1$). The time required for computing the value of each cell of the table will be clearly polynomial in terms of the input, so the claim regarding the computational
complexity of the procedure is immediate.

We begin with examining the edge between $v$ and its first predecessor. Hence, we are focusing on the subgraph $T_v^1;$ for the relevant definition we refer to the proof of \cref{dp-maxlength}. So, we should pick $v$ as a casting voter since it is a sink, meaning a vertex of zero out-degree, in the considered subgraph. There are two possible options regarding the considered edge $(v_1,v)$: either use it by making $v_1$ to delegate to $v$ or not, i.e., force $v_1$ to cast a ballot as well. Among the two, we pick the one of minimal cost that is feasible. 

Regarding the first case, asking voter $v_1$ to delegate to $v$, then we must also ensure that the voting power of $v_1$ is at most $k-1,$ in order to make $v$ vote with a power of at most $k$. Notice that by the fact that we have not yet considered any in-coming neighbors of $v$ other than $v_1,$ the voting power of $v$ only comes from voters in $T_v^1$.
Then, using a reasoning similar to the one presented in the proof of \cref{dp-maxlength}, we obtain the following relation:
\begin{align*}
    dp[v,1,k]=v(v)+d(v_1)-v(v_1)+dp[v_1,p(v_1),k-1].
\end{align*}
On the other hand, if we ask both $v$ and $v_1$ to cast a ballot, then we should ensure that the voting weight of $v_1$ is no more than $\ell$ to be feasible. Again, in analogy to \cref{dp-maxlength}, the following holds: \begin{align*}
    dp[v,1,k]=v(v)+dp[v_1,p(v_1),\ell].
\end{align*}

We now move to the case $i> 1.$
We will consider two possible choices when examining the edge between $v$ and $v_{i+1}$: use the edge from $v_{i+1}$ to $v$ or not. As before, among the two we pick the cheapest feasible option. First, say that we use the edge from $v_{i+1}$ to $v$, i.e. that $v_{i+1}$ will delegate to $v$. Note that while computing $dp[v,i,\cdot]$ we have already taken into consideration the fact that $v$ should be a casting voter.
To bound the voting power of $v$ by at most $k$, for a fixed value of $k$, we should bound the voting power of $v$ from $T_v^i,$ by a number $x\leq k$ and the voting power of $v_{i+1}$ by $k-x$. Additionally, we need to add the delegating cost of $v_{i+1}$ as the corresponding voter will delegate to $v$ and also subtract their voting cost, as this has been counted in $dp[v_{i+1},p(v_{i+1}),\cdot]$ since $v_{i+1}$ was a source in the subgraph considered for the computation of the value appearing in the aforementioned cell. Then, the following holds: 
\begin{align*}
   &dp[v,i+1,k]=\\
   &min_{x\in \{1,2,...,k-1\}} \{dp[v,i,x]+dp[v_{i+1},p(v_{i+1}),k-x]\}+d(v_{i+1})-v(v_{i+1}). 
\end{align*}

Next, suppose that we don't use the edge from $v_{i+1}$ to $v$, i.e., that $v_{i+1}$ will be finally asked to cast a ballot. Then we need to ensure that the voting power of $v_{i+1}$ is at most $\ell$ in order to have a feasible solution. Obviously, we don't need to add any delegating costs. Moreover, the voting costs of $v$ and $v_{i+1}$ have already been counted in $dp[v,i,\cdot]$ and $dp[v_{i+1},p(v_{i+1}),\cdot]$ respectively, so we have the following: \begin{align*}
dp[v,i+1,k]=dp[v,i,k]+dp[v_{i+1},p(v_{i+1}),\ell].
\end{align*}

Finally, we also set $dp[v,0,k]=v(v)$ if $v$ is a leaf, for any $k\geq 1$ and we also note that the cost of the cost-minimizing solution lies in $dp[x,p(x),\ell];$ a backtracking procedure can recover the solution itself.

\subsection{Proof of \Cref{power-hard}}
For clarity, we firstly prove the hardness of \textsc{Bounded Power} for $\ell=4$ and at the end we generalize the construction to work for larger $\ell$. We prove the hardness of \textsc{Bounded Power} by reduction from \textsc{Vertex Cover}. There, for a given graph of $\nu$ vertices and $\mu$ edges and a parameter $k$, the goal is to select $k$ out of the $\nu$ vertices of the graph covering all of its edges. Consider an instance $\mathcal{I}$ of \textsc{Vertex Cover} on 3-regular graphs, we create the following encoding of $\mathcal{I}$ into an instance of \textsc{Bounded Power} problem: 
First, for every vertex of $\mathcal{I}$, we construct a \emph{vertex voter} corresponding to $v$. Moreover, for every edge $e$ of $\mathcal{I}$ we construct an \emph{edge voter} corresponding to $e$. Furthermore, if an edge $e$ of $\mathcal{I}$ is between vertices $i,j$, we construct edges, i.e., possible delegations, from the voter corresponding to $e$ to the voters who correspond to $i$ and to $j$. We finally add a set $D$ of $\nu$ \emph{dummy voters}, who will vote at a cost of $0$. For each vertex voter we add an edge to exactly one dummy voter in a way that each voter from $D$ has in-degree exactly 1 from vertex voters. Finally, for each voter in $D$ we add two additional voters, who are willing to delegate to them.
We set the rest voting costs to $1$ and all delegating costs to $0$ and also we set $\beta=k$ and $\ell=4$. Note that $\Delta=2$.

Suppose that $\mathcal{I}$ is a positive instance of \textsc{Vertex Cover}. Then, consider a delegation function that forces the vertex voters corresponding to vertices of the cover as well as the voters from $D$ to cast a ballot. The cost of this delegation function is no more than $k$. Under it, edge voters can all be represented by the vertex voters that have been selected to cast a ballot while the vertex voters that have not been selected to cast a ballot can be represented by their corresponding dummy voters. To prove feasibility of the suggested delegation function it remains to show that the voting power of each casting voter is bounded by $\ell$. Notice that the voting power of voters in $D$ will be at most $4$: exactly $4$ for those who represent vertex voters and $3$ for the rest (as they represent themselves together with the added $2$ voters that are only willing to delegate to them). Regarding casting vertex voters, those will have a voting power of at most $4$, representing themselves as well as at most $3$ edge voters each, by the fact that the graph in $\mathcal{I}$ is 3-regular.

For the reverse direction, suppose that $\mathcal{I}$ is a negative instance of \textsc{Vertex Cover}. Then, due to the budget constraint, at most $k$ voters from $V$ can be selected to cast a ballot. Since those do not correspond to a cover in $\mathcal{I}$, at least one edge voter, say $e$, will be represented by a voter from $D$ that will cast a ballot. This breaks the upper bound of $\ell=4$ on the voting power of casting voters as there will be a dummy voter that will also represent themself, the added $2$ voters that can only delegate to them, as well as the vertex voter that appears in the path from $e$ to them. Therefore no feasible delegation function exists, proving that the encoding of $\mathcal{I}$ is a negative instance of \textsc{Bounded Power}.

Increasing the number of voters approving each dummy voter from $2$ to $\ell-2$, and applying the same arguments, proves the theorem for any value of the parameter $\ell$ and completes the proof.

\section*{Missing Proofs from \Cref{sec:control}}

\subsection{Proof of \Cref{control-hard1}}
We will prove hardness for the problem of adding $k$ voters, from a predefined set, to make a designated voter $x$ the unique super-voter under every feasible cost-minimizing delegation function. We will reduce from \textsc{Vertex Cover}. Given an instance of \textsc{Vertex Cover} $(G,k)$ we create an instance $\mathcal{I}$ of the considered election control problem that contains the following: 
    \begin{itemize}
\item A set of voters $E$ containing one voter of delegating cost equal to $0$ and voting cost equal to $1$, for each edge in $G$.
\item A set of voters $V$ containing one voter of delegating cost equal to $0$ and voting cost equal to $1$, for each vertex in $G$. All voters in $V$ are considered being unregistered, meaning that between those we should choose some to add in the election. 
\item A set $D$ of $k+m-1$ voters, each of delegating cost equal to $0$ and voting cost equal to $1$.
\item Two special voters $x$ and $y$ of voting cost equal to $0$ and delegating cost equal to $0$. 
\item For each edge $e$ that is incident to vertices $u$ and $v$ in $G$, we add two possible delegations: one from the voter of $E$ that corresponds to $e$ towards the voter of $V$ that corresponds to $v$ and one from the voter of $E$ that corresponds to $e$ towards the voter of $V$ that corresponds to $u$. 
\item There is an edge from each voter in $V$ to voter $x$.
\item There is an edge from each voter in $D$ to voter $y$. 
\end{itemize}

To prove the forward direction, suppose that there is a vertex cover of size $k$ in $G$. Adding to the election the $k$ voters from $V$ that correspond to vertices from the vertex cover has the following implications. Every each of those $k$ voters has a unique path (of unit length) to $x$. Additionally, each voter in $E$ has a path to $x$, via an added voter from $V$. Voters $x$ and $y$ will be casting voters as they have less voting that delegating cost. On the other hand, in every cost minimizing delegation function, each voter from $V, E$ and $D$ that has a path to a casting voter will delegate because their delegating cost is less than their voting one. Voters from $D$, and only those, will delegate to $y$, under every cost-minimizing delegation function. As a result, the voting power of $y$ will be $1+k+m-1$ since $y$ represents herself as well as all voters from $D$, whereas the voting power of $x$ will be $1+k+m$ as $x$ represents all other voters. This shows that $x$ is the sole super-voter under every cost-minimizing delegation function and completes the forward direction of the proof.

For the reverse direction, suppose that $G$ doesn't have a vertex cover of size $k$. We will show that there is no way to add $k$ voters from $V$ in $\mathcal{I}$ so as to make $x$ the sole super-voter under every cost-minimizing delegation function. Pick any arbitrary set of $k$ voters from $V$ to add in the election. Any cost-minimizing delegation rule would ask voters from $D$ to delegate to $y$. Hence the voting power of $y$ is exactly $k+m$ under every delegation function. Furthermore, any cost-minimizing delegation rule would ask the selected $k$ voters from $V$ to delegate to $x$, as well as all voters from $E$ that do have a path to $x$. Since the selected $k$ voters do not correspond to a vertex cover in $G$, at least one edge of $G$ doesn't have one of its endpoints in the selected set. Then, for at least one of the voters from $E$ it holds that both of her approved voters aren't selected to be included in the election. As a result, the out-degree of the corresponding voter equals $0$ and she should cast a vote in every feasible cost-minimizing delegation function. Hence, at most $m-1$ voters from $E$ will delegate to $x$ and, in turn, her voting power will be upper bounded by $1+k+m-1$. Consequently, $x$ will not be the unique super-voter.

\subsection{Proof of \Cref{control-hard2}} 
We will show hardness for the problem of deleting $k$ voters, to make a designated voter $x$ the unique super-voter under every feasible cost-minimizing delegation function. To prove the statement we will reduce from \textsc{Clique}. Given an instance of \textsc{Clique} $(G,k)$ we create an instance $\mathcal{I'}$ of the considered election control problem that contains the following: 
    \begin{itemize}
\item A set of voters $E$ containing one voter of delegating cost equal to $0$ and voting cost equal to $1$, for each edge in $G$.
\item A set of voters $V$ containing one voter of delegating cost equal to $0$ and voting cost equal to $1$, for each vertex in $G$. 
\item A set $D$ of $(n-k)+(m-\binom{k}{2})+1$ voters, each of delegating cost equal to $0$ and voting cost equal to $1$.
\item A set of special voters $x$ and $y_1,y_2,\dots,y_k,y_{k+1}$ of voting cost equal to $0$ and delegating cost equal to $0$. 
\item For each edge $e$ that is incident to vertices $u$ and $v$ in $G$, we add two possible delegations: one from the voter of $E$ that corresponds to $e$ towards the voter of $V$ that corresponds to $v$ and one from the voter of $E$ that corresponds to $e$ towards the voter of $V$ that corresponds to $u$. 
\item There is an edge from each voter in $V$ to every voter in $\{y_1,y_2,\dots,y_{k+1}\}$.
\item There is an edge from each voter in $D$ to voter $x$. 
\end{itemize}

We call $Y$ the set $\{y_1,y_2,\dots,y_{k+1}\}$. To prove the forward direction, suppose that there is a clique of size $k$ in $G$. Deleting from the election the $k$ voters from $V$ that correspond to vertices from the clique has the following implications. Every each of the remaining $n-k$ voters from $V$ has a unique path (of unit length) to each vertex from $Y$. Additionally, there are $m-\binom{k}{2}$ voters from $E$ having a path to each vertex from $Y$, via a voter from $V$. The rest $\binom{k}{2}$ voters from $E$ correspond to edges of the clique and hence both voters from $V$ that correspond to the endpoints of the analogous edges have been deleted. As their out-degree equals $0$ they have to cast a ballot under every feasible cost-minimizing delegation function. Moreover, voters $x$ and voters in $Y$ will be casting voters as well as they have less voting that delegating cost. On the other hand, as before, in every cost minimizing delegation function, each voter from $V, E$ and $D$ that has a path to a casting voter will delegate because their delegating cost is less than their voting one. Voters from $D$, and only those, will delegate to $x$, under every cost-minimizing delegation function. As a result, the voting power of $x$ will be $1+(n-k)+(m-\binom{k}{2})+1$ since $x$ represents herself as well as all voters from $D$, whereas the voting power of any voter in $Y$ will be at most $1+(n-k)+(m-\binom{k}{2})$ under any delegation function as voters in this set represent all other delegating voters. This shows that $x$ is the sole super-voter under every cost-minimizing delegation function and completes the forward direction of the proof.

For the reverse direction, suppose that $G$ doesn't have a clique of size $k$. We will show that there is no way to delete $k$ voters from $V$ in $\mathcal{I}$ so as to make $x$ the sole super-voter under every cost-minimizing delegation function. Before that, we highlight that in order to make $x$ the unique super-voter it only makes sense to attempt deleting voters that have paths towards voters in $Y$. Additionally, it is without loss of generality to assume that no voter from $Y$ will be deleted; this is because the voters in that set are more than $k,$ and they all have the same predecessors, so deleting only a subset of them will not be helpful towards making $x$ the unique super-voter. Moreover, it is always better to delete from $V$ rather than from $E$, as a deletion of a voter from $E$ can cause strictly less decrease in the voting power of a voter in $Y$ than the deletion of its successor in $V$. Pick any arbitrary set of $k$ voters from $V$ to add in the election. Any cost-minimizing delegation rule would ask voters from $D$ to delegate to $x$. Hence the voting power of $x$ is exactly $1+(n-k)+(m-\binom{k}{2})+1$ under every delegation function. Furthermore, any cost-minimizing delegation rule would ask the remaining $n-k$ voters from $V$ to delegate to a vertex from $Y$, as well as all voters from $E$ that do have a path to a vertex in that set. Since the deleted $k$ voters do not correspond to a clique in $G$, at most $\binom{k}{2}-1$ edges of $G$ have both of its endpoints in the selected set. Then, for at most $\binom{k}{2}-1$ voters from $E$ it holds that both of her approved voters are deleted from the election and as a result, the remaining $m-\binom{k}{2}+1$ voters will delegate to a vertex from $Y$, under at least one feasible cost-minimizing delegation function. Hence, the voting power of a voter in $Y$ will be at least $1+(n-k)+(m-\binom{k}{2}+1)$. Consequently, $x$ will not be the unique super-voter under at least one delegation function.

\subsection{Proof of \Cref{control_pos1}}
In the examined problem, one should first ensure that the designated voter $x$ is indeed a casting voter before working towards making her a super-voter. Trivially, if $x$ is yet unregistered, she should definitely be one of the voters to be added, so it is necessary to include $x$ by reducing the available budget of $k$ voters by one. So, we assume that $x$ is indeed a registered voter. On top of that, we need to check if $x$ is a casting voter in every cost-minimizing delegation function that satisfies the restrictions of \textsc{Delegate Reachability}, in the original instance containing only registered voters. This can be done in polynomial time by \cref{casting_for_all} (appearing in the proof of \cref{control:pos2}). If $x$ doesn't satisfy this condition, then the inclusion of additional voters in the electorate cannot make her a casting one, so the instance is a negative one. 

 Now, suppose that $x$ is a casting voter in every feasible cost-minimizing delegation function. Then, it is safe to delete any out-going edges of $x$ making $x$ the root of an upwards-directed tree, say $T$. The goal then is to make her the only super-voter, by adding unregistered voters to the electorate. As $x$ is a casting voter, there is a set of voters who will delegate to $x$ in every cost minimizing delegation function. Those vertices can be determined as done in \cref{control:pos2}, and have been called $D(x)$ there. Among the unregistered voters we only add voters with voting cost strictly larger than delegating cost and with a path towards a vertex from $D(x)$ (either of unit length or that goes through other unregistered voters), as only such voters will necessarily delegate to $x$ under every delegation function and hence are able to increase the voting power of $x$. We call those voters $T_x$. We suggest the greedy algorithm that adds unregistered voters by layers: first adds as many as possible from $T_x$ that have a direct edge to a vertex from $D(x)$, then adds as many as possible from the remaining in $T_x$ that have a direct edge to a voter added in the previous step, and continues iteratively until adding $k$ voters. The greedy strategy is sufficient as there is no advantage from adding a voter that has a path to $x$ through unregistered voters, unless all of those unregistered have already been added.

\subsection{Proof of \Cref{control:pos2}}
The procedure that we are going to use consists of three parts: (a) ensure that $x$ is a casting voter in every cost minimizing feasible, under the reachability constraint, delegation function, (b) compute the lower bound of the voting power that $x$, as a casting voter, has in each such delegation function, (c) delete appropriate voters to make $x$ the sole super-voter under every such. The first part is necessary since if $x$ is delegating in at least one cost-minimizing delegation function, then he cannot be the sole super-voter under every such. The second is a necessary component for solving the third part, which resolves the computational problem of interest.

We begin by noticing that if $x$ is not a casting voter under every cost-minimizing delegation function then we can make him such by deleting the voter that $x$ approves for being her representative, and in turn also reducing the available budget $k$ by one. But first, we need to determine whether this step is necessary to be done. In a cycle $\mathcal{C}$ of a (weakly) connected component of a delegation graph that satisfies $\Delta=1$ such that all vertices of which have a voting cost no less than their delegating cost, we call a vertex $u$ as \emph{forced voter} if it is the (only) one minimizing the following expression: $v(u)+\sum_{i\in \mathcal{C}\setminus \{u\}}d(i)$. To justify this, we remind that in every such a cycle, all voters should be represented, which is that exactly one of them should cast a ballot in every cost-minimizing delegation function. The total cost incurred in the solution by that cycle is exactly $v(u)+\sum_{i\in \mathcal{C}\setminus \{u\}}d(i)$, if the vertex $u$ is chosen to cast a ballot.

\begin{claim}
\label{casting_for_all}
Given a delegation graph $G$, a vertex $x$ corresponds to a casting voter in every cost-minimizing, feasible under the reachability constraint delegation function on $G$ if at least one of the following holds: (a) $x$ is a root of an upwards-directed tree, (b) $x$ is the forced voter of a cycle of a (weakly) connected component of $G$, (c) $vc(x)<dc(x)$.
\end{claim}
\begin{proof}
    If any of the conditions (a), (b) or (c) holds, it is easy to verify that any cost-minimizing feasible delegation function would ask $x$ to cast a ballot. To prove the forward direction, suppose that $x$ will cast a ballot, while all (a), (b) and (c) do not hold. If $x$ is not a root then either it belongs to a cycle or there is a path from $x$ to some other vertex $x'$ from which there is no path back to $x$, in which case $x'$ is already represented in the solution (by either herself or another vertex) resulting to a representation of $x$ as well, so there is no need of picking $x$ as a casting voter as well, provided that $vc(x)\geq dc(x)$. So, say that $x$ belongs to a cycle, without being the forced voter. But picking the forced voter, and only this, from the cycle, is sufficient to cover all of those. Hence, we have a contradiction.
\end{proof}

Then, using the conditions from \cref{casting_for_all} we can in polynomial time determine whether $x$ is a casting voter in all delegation functions that achieve minimization of cost or not, in which case we can make her one by deleting her outgoing edge.

So, now we can suppose that $x$ is a casting voter, therefore we consider $x$ as being a root of an upwards-directed tree. We now turn to computing her voting power $p(x)$. The vertices in the same (weakly) connected component as $x$, say $T_x$ will delegate to $x$ in at least one delegation function that is cost-minimizing. Specifically, those of delegating cost lower than voting cost will delegate to $x$ under every cost-minimizing delegation function and voters for which those costs are equal will delegate in some functions. No voter in this component has strictly less voting cost than delegating cost as those voters would vote themselves. Consider a virtual deletion of all those vertices in $T_v$ that will cast a ballot in at least one cost-minimizing delegation function. The vertices that remain having a path to $x,$ will delegate to her under every such function, so their cardinality (plus one) will be the minimum voting power that $x$ will have in any cost minimizing delegation function, in a tight sense, meaning that there is at least one such delegation function under which the voting power of $x$ exactly equals the computed number. We call those vertices $D(x)$.

Then, we determine the set $S$ of vertices for each of which there are at least $p(x)-1$ vertices having a directed path to them. All vertices of $S$ should either be deleted or, at least, their voting power should be reduced, which can be done by deleting some of their predecessors. Pick a voter $y\in S$. If $y$ has no predecessors in $S$ we delete $y$ itself. If $y$ has predecessors in $S,$ it is safe to start by deleting one which has the path of maximal length to $y$, say $z$. This holds because the voting power of $z$ should be reduced and this can be achieved by deleting $z$ or one of her predecessors $P(z)$, but since no of her predecessors are in $S$, it is always better to delete $z$ as it is the one in $P(z)\cup\{z\}$ that is closer to $y$, meaning that it might reduce the voting power of $y$ or of other vertices in $S$ between $z$ and $y$. After deleting $z,$ we recompute the set $S$ and repeat the described procedure until we have deleted $k$ vertices or until the newly computed set $S$ is empty.

\end{document}